\documentclass{article}
\usepackage{natbib}
\setcitestyle{round}
\usepackage[affil-it]{authblk}
\usepackage[english]{babel}

\usepackage[letterpaper,top=2cm,bottom=2cm,left=3cm,right=3cm,marginparwidth=1.75cm]{geometry}

\usepackage{graphicx,amsmath,amsfonts,amsthm}
\usepackage{graphicx}
\usepackage{booktabs}
\usepackage{tikz}
\usepackage{qtree}
\usepackage{tikz-qtree}
\usepackage[colorlinks=true, allcolors=blue]{hyperref}
\usepackage{xcolor}
\usepackage{caption,subcaption}

\newtheorem{prop}{Proposition}
\newtheorem{lemma}{Lemma}
\newtheorem{corollary}{Corollary}
\newtheorem{definition}{Definition}

\newtheorem{theo}{Theorem}

\newtheorem{conj}[theo]{Conjecture}

\newtheorem{quest}[theo]{Question}

\title{Bipartite peak-pit domains}

\date{}

\author[1,2]{Alexander Karpov\footnote{Authors are  listed in alphabetical order}\thanks{Correspondent author: akarpov@hse.ru}} 
\author[3]{Klas Markstr{\"o}m} 
\author[4]{S{\o}ren Riis}
\author[4]{Bei Zhou}
\affil[1]{HSE University, Moscow, Russia}
\affil[2]{Institute of Control Sciences, Russian Academy of Sciences, Moscow, Russia}
\affil[3]{Ume\aa\ University} 
\affil[4]{Queen Mary University of London}

\begin{document}
\maketitle

\begin{abstract}
In this paper, we introduce the class of bipartite peak-pit domains. This is a class of Condorcet domains which include both the classical single-peaked and single-dipped domains. Our class of domains can be used to model situations where some alternatives are ranked based on a most preferred location on a societal axis, and some are ranked based on a least preferred location. This makes it possible to model situations where agents have different rationales for their ranking depending on which of two subclasses of the alternatives one is considering belong to.  The class of bipartite peak-pit domains includes most peak-pit domains for $n\leq 7$ alternatives, and the largest Condorcet domains for each $n\leq 8$.

In order to study the maximum possible size of a bipartite peak-pit domain we introduce set-alternating schemes. This is a method for constructing well-structured peak-pit domains which are copious and connected.  We show that domains based on these schemes always have size at least $2^{n-1}$ and some of them have sizes larger than the domains of Fishburn's alternating scheme. We show that the maximum domain size for sufficiently high $n$ exceeds $2.1973^n$. This improves the previous lower bound for peak-pit domains $2.1890^n$ from \cite{karpov2023constructing}, which was also the highest asymptotic lower bound for the size of the largest Condorcet domains. 
\end{abstract}

\section{Introduction}
The simplest non-trivial  model for the possible rankings of a collection of alternatives is perhaps the single-peaked domain. This domain was introduced formally in \cite{black} and corresponds to a society in which all voters have a common political axis, each voter has a  preferred position on this axis and ranks alternatives according to how close to this preferred position they are. \cite{black} showed that on any domain of this type majority voting will select a unique winner, and in fact produce a transitive ranking of all  alternatives.  In terms of modelling the behaviour of a group of voters the single-peaked domain provides a sensible first approximation of, for example, a society where opinions correlate with a left-right, or a liberal-conservative axis.  It is however clearly not a universal model. A second distinct model is given by the single-dipped domains, in which voters instead have a least preferred position on a common axis and rank alternatives higher when they are further away from this position.  Up to relabeling of alternatives there is  exactly one maximal single-peaked domain for $n$ alternatives, known as Black's single-peaked domain, and one maximal singe-dipped domain. Both have size $2^{n-1}$, where $n$ is the number of alternatives, and are in fact each other's duals, i.e. each order in one appears reversed in the other.     

In real life situations the collection of opinions is often not strictly determined by a common political axis and  modeling requires more flexible options than Black's single-peaked domain, and the single-dipped domain.  \cite{arrow63} introduced one such class by requiring that the restriction to any triple of alternatives must be single-peaked in Black's sense, but letting go of the common political  axis. These domains are now known as Arrow's single-peaked domains  and have been studied in great detail. \cite{slinko2019condorcet} showed that maximal domains of this type have size $2^{n-1}$, just as Black's, and they are minimally rich, meaning that every alternative is ranked first by some order.   While this class generalises the single-peaked domain greatly it does not include the single-dipped domain.   In the theory of computational social choice an alternative approach to expanding the set of domains has been to investigate if a domain is in fact, for a suitable distance measure, close to being single-peaked. When this is the case one then can attempt to exploit the nearby single-peaked domain to simplify computational problems which are known to have efficient algorithms on single-peaked domains. Examples of this approach can be found in \citep{Bredereck}.

One common generalisation of the single-peaked and single-dipped domains is given by the peak-pit domains. These domains follow the set-up given by Arrow and require the restriction to each triple of alternatives to be either single-peaked or single-dipped in Black's sense.  The class of peak-pit domains  is much larger and structurally diverse than the  Arrow's single-peaked domains, giving the modeller the ability to match a wider range of groups or societies. However, this diversity also means that general peak-pit domains do not have the unified social interpretation that the single-peaked domains have. Peak-pit domains have been studied in numerous papers, see e.g \cite{DKK:2012,Li2021},  and within the class one finds interesting trade-offs between different forms of diversity.   It is easy to show that if a peak-pit domain is minimally rich then it is an Arrow's single-peaked domain. However, if we consider domains which are not minimally rich  we find  examples which are much larger than Arrow's domains,  thereby giving more diversity in terms of the total number of distinct opinions.  The first such examples were found by \cite{fishburn1996acyclic}, and were only larger than Arrow's domains by a linear factor, but by now several constructions \cite{fishburn2002acyclic,karpov2023constructing} have given exponentially larger domains. Most such examples have been relatively complicated to describe and analyse.

Another common generalisation has been given in terms of domains which are unions of single-peaked and single-dipped domains. In these domains each voter either has single-dipped or single-peaked preferences, and the domain may not be a Condorcet domain.  For the union of the maximal single-peaked and single-dipped domains \cite{achuthankutty2018dictatorship,BERGA200039} showed that the conclusion of the Gibbard-Satterwhaite theorem remains valid, saying that strategy-proof voting rules must be dictatorial. Later works \cite{feigenbaum2015strategyproof,ALCALDEUNZU201821,alcaldeunzu2023strategyproofness} have shown that by adding additional public information about the voters one can derive strategy-proof non-dictatorial decision rules for domains in this class. The main limitation of these models is that each voter must belong to one of two quite restrictive classes.  Another generalisation is given in  \cite{Yang} which consider domains which are the union of several single-peaked domains, each with a separate axis. Here the case of two distinct axes was found to remain well-behaved for several computational social choice problems. 

In this paper we will introduce several sub-classes of the peak-pit domains which mix the structures of  single-peaked and single-dipped domains in a new well-structured way. Instead of dividing the voters into two classes we divide the alternatives into two classes, which informally can be seen as attractive and repulsive alternatives respectively. The most general of our classes are the \emph{bipartite} peak-pit domains. A peak-pit domain is  bipartite if  the set of alternatives can partitioned into two parts, such that on the first part the domain is single-peaked and on the second part the domain is single-dipped.  This type of domain captures situations where voters have different rationales for their ranking of the two classes of alternatives. Within the first class they have a most preferred location on a societal axis, and within the second class they have a least preferred location. Here we impose no additional restrictions on the ranking of pairs of alternatives from both classes. As we shall show, most peak-pit domains on $n\leq 7$ alternatives are bipartite, as are the largest peak-pit domains on $n\leq 8$ alternatives. 

Next we introduce a subclass of the bipartite peak-pit domains called \emph{midpoint bipartite} peak-pit domains. Here alternatives from the two parts are intermixed by saying that the restriction to a triple whose midpoint belongs to the first of the two parts, shall be single-peaked, and vice versa.  Equivalently,  we can think of these domains as describing a society in which voters rank triples in a single-peaked or a single-dipped way, and they have one set of local midpoints which they prefer to give high local rank and one set of local midpoints which they prefer to give low local rank.  Up to $n=7$ all maximum Condorcet domains belong to this class, but most peak-pit domains do not.

Finally, in order to simplify the structure of our domains further, and allow us to give a lower bound for how large bipartite peak-pit domains can become,  we restrict our focus to domains with just two distinct never conditions, rather than the full set allowed for peak-pit domains. These are the \emph{set-alternating} domains. As we will show, this class contains domains larger than any previously known, both when seen as peak-pit domains and more generally as Condorcet domains. 

\subsection{Related literature}
The search for large well-behaved domains has a long history. \cite{AJ84} showed that there are Condorcet domains, i.e. domains on which pairwise majority voting leads to a transitive ranking of the candidates, which are larger than $2^{n-1}$, the size of Black's single-peaked domain, and asked how large such domains can be.  This question was reiterated in the survey  \citet{kim1992overview}, which listed  several fundamental unsolved problems in mathematical social science. The first of which is "What is the largest size of a set of linear preference orders for $n$ alternatives such that majority voting is transitive when each voter chooses his preferences from this set?". More recent surveys can be found in \citep{elkind2022,karpov_a,puppe2023maximal}).   

\cite{Johnson78}, and later  \cite{Craven1992}, conjectured that a Condorcet domain of maximum size contains at most $2^{n-1}$  linear orders. Here $2^{n-1}$ is the maximum cardinality in  several well-known domain classes, such as the single-peaked, single-dipped, and group-separable domains. After Craven's paper, a new counterexample on five alternatives to the conjecture was given in \citet{PF:1992}. The examples in \cite{AJ84} were not well known at the time.  \cite{fishburn1996acyclic} generalized the new counterexample by introducing an alternating scheme resulting in what is now called Fishburn's domains. These domains mix never-top and  never-bottom triples and has a size of order $n\times2^n$ \citep{GR:2008}.

\cite{fishburn1996acyclic} conjectured that among Condorcet domains such that for any triple of alternatives do not satisfy a never-middle condition the alternating scheme provides domains of maximum cardinality.

The conjecture posed by Fishburn is true for three to five alternatives \citep{fishburn1996acyclic}, and six alternatives \cite{fishburn2002acyclic}. In \citep{GR:2008} it was claimed to be true for seven alternatives, though without including a proof. This was later proven independently in \cite{n7paper}.  For eight alternatives there is a counterexample to Fishburn's conjecture \citep{8Alternatives}.

For a sufficiently large $n$ the lower bound on the size of maximum Condorcet domains has been increased to $2.1708^n$ \citep{fishburn1996acyclic} and later to  $2.1890^n$ \cite{karpov2023constructing}.

In this paper we present a sequence of Condorcet domains that for a sufficiently large $n$ have a size exceeding $2.1973^n$, improving the existing lower bounds. The value of this result lies not only in improving the lower bound but also in presenting a well-structured domain that does not depend on recursive constructions. Previous cited lower bounds, were based on the replacement scheme that was iteratively applied, starting with a carefully chosen domain on a small number of alternatives.

\subsection{Overview of the paper}
The structure of the paper is as follows. Section \ref{sec:prelim} contains notation and definitions used in the paper. We present Bipartite peak-pit domains in Section \ref{sec:bipar}. Section \ref{sec:set_alter} introduces set-alternating schemes and discusses the structure of the domains these generate. Section \ref{sec:asymp} contains the analysis of the asymptotic growth rate of some of the largest domains in this class.  We conclude in section \ref{sec:disc} with some additional comments and open problems.

\section{Preliminaries}
\label{sec:prelim}

Let a finite set  $X=[n]=\{1, \dots, n\}$ be the set of alternatives. Let $L(X)$ be the set of all linear orders over $X$.  Each agent $i \in N$ has a preference order $P_i$ over $X$ (each preference order is a linear order). For brevity, we will write preference orders as  strings, e.g. $12 \dots n$ means $1$ is the best alternative, $n$ is the worst.

A subset of preference orders $D\subseteq L(X)$ is called a \emph{domain} of preference orders. A domain $D$ is a \emph{Condorcet domain} if whenever the preferences of all agents belong to the domain, the majority relation of any preference profile with an odd number of agents is transitive. A Condorcet domain $D$ is \emph{maximal} if every Condorcet domain $D'\supseteq D$ (on the same set of alternatives) coincides with $D$. In our paper we are focused on maximal Condorcet domains	unless otherwise specified.

The \emph{dual} of a partial order $R_1$ is the partial order $R_2$ for which $x R_2 y$ if and only if $y R_1 x$.  The dual of a domain $D$ is the domain consisting of the dual of each order in $D$.

A \emph{societal axis} for a domain is a designated linear order on the set of alternatives. This is usually assumed to correspond to some organising principle for the alternatives, e.g. a left-right political scale.  The societal axis is not necessarily an element of the domain, e.g. for some domains which are not maximal domains. 

\cite{Sen1966} proved that  a domain is a Condorcet domain if the restriction of the domain to any triple of alternatives $(a,b,c)$ satisfies a never condition. A never condition can be of three forms $xNb$, $xNm$ $xNt$, referred to as a never bottom, a never middle, and a never top condition respectively. Here $x$ is an alternative from the   triple and  $xnb$, $xNm$, and $xNt$ means that $x$ is not ranked last, second, or  first respectively in the restricted domain. Fishburn noted that for domains with a societal axis never conditions can instead be described as  $iNj$, $i,j\in [3]$. $iNj$ means that $i^{th}$ alternative from the triple according to societal axis does not fill in $j^{th}$ place within this triple in each order from the domain. For example restriction $abc$ to triple $a,b,c\in[n]$, $a<b<c$ satisfies never conditions $1N2,1N3,2N1,2N3,3N1,3N2$, but violates never conditions $1N1,2N2,3N3$.

A Condorcet domain $D$  is \emph{connected} if, given any two orders from the domain, the second order can be obtained from the first  by a sequence of transpositions of neighbouring alternatives such that all orders generated by this sequence belong to the domain. A Condorcet domain has \emph{maximal width} if it contains a pair of completely reversed linear orders. A Condorcet domain $D$ is \emph{unitary} if it contains order $123\ldots n$.

Following \cite{slinko2019condorcet} we say that a Condorcet domain is \emph{copious} if the restriction to any triple has size 4, or equivalently every triple, satisfies a single never condition. A domain is \emph{ample} if the restriction to any pair has size two. Every copious Condorcet domain is ample. 

A domain which satisfies a never condition of the form $xN3$ for every triple is called \emph{Arrow's single-peaked domain}.

A domain $D$ is a \emph{peak-pit} domain if, for each triple of alternatives, the restriction of the domain to this triple is either single-peaked ($iN3$  restriction), or single-dipped ($iN1$ restriction).

A domain $D$ is called a \emph{Fishburn domain} if it satisfies the {\em alternating scheme} \cite{fishburn1996acyclic}: there exists a linear ordering of alternatives $a_1, \ldots, a_m$ such that for all $i, j, k$ with $1\le i<j<k\le m$ the restriction of the domain to the set $\{a_i,a_j,a_k\}$ is single-peaked with axis $a_ia_ja_k$ if $j$ if is even (odd), and it is single-dipped with axis $a_ia_ja_k$ if $j$ is odd (even). Note that there is either one or two Fishburn domains depending on the parity of $n$. The dual of a Fishburn domain is also a Fishburn domain.

Recall that the asymptotic notation $f(n)\in \Omega(g(n))$ means that $\lim_{n\rightarrow \infty}{\frac{f(n)}{g(n)}}>0$,

\cite{fishburn1996acyclic} introduced the function
\[
f(n)=\max \{|D|: \text{$D$ is a Condorcet domain on a set of $n$ alternatives}\}.
\]
Similarly \cite{karpov2023constructing} introduced the function
\[h(n)=\max \{|D| : D\text{ is a peak-pit Condorcet domain on a set of $n$ alternatives}\}.
\]
and  also showed that $f(n)\in \Omega(2.1890^n)$  and $h(n)\in \Omega(2.1045^n)$.

\begin{definition}
Let $D_1$ and $D_2$ be two domains of equal cardinality on sets of alternatives $X_1$ and $X_2$, respectively. 
We say that domains $D_1$ and $D_2$ are {\em isomorphic} if there are a bijection $\psi\colon X_1\to X_2$ and a bijection $\sigma \colon D_1\to D_2$  such that for each $x \in D_1$  we have $\sigma(x)=x^{\psi}$, where $x^{\psi}$ is an order with permuted alternatives according to $\psi$. If $D_1$ and $D_2$ are isomorphic, we write $D_1\cong D_2$.
\end{definition}

Every domain $D$ is associated with a graph $G_D$ \citep{puppe2019condorcet}. The set of linear orders from $D$ is the set of vertices  $V_D$ and for two orders $u,w\in D$ we draw an edge between them if $D$ does not contain order $x$ that is between $u,w$ in terms of Kemeny betweeness (which means that $x$ agrees with all binary comparisons on which $u,v$ agrees). \cite{puppe2019condorcet} proved that for each Condorcet domain $D$ the graph $G_D$  is a median graph (as defined by \citep{mulder}). If Condorcet domain $D$ is connected, then in the associated domain each edge represents one swap of neighbouring alternatives. 

For set complements we  will use the notation  $\overline{A}=[n]\setminus A$, where $A\subseteq [n]$. 

\section{Bipartite peak-pit domains}
\label{sec:bipar}

We will start by formalising some of the domain types described in the introduction.

\begin{definition}
    A peak-pit domain  $D$ is a \emph{bipartite peak-pit domain} if there exists a subset $A$ of the alternatives such that the restriction of $D$ to $A$ is Arrow's single-peaked and the restriction of $D$ to $\bar{A}$ is the dual of an Arrow's single-peaked domain. 
\end{definition}
The class of bipartite peak-pit domains is closed under isomorphism so every domain of type is isomorphic to one where $A=\{1,2,\ldots,k\}$, $\bar{A}=\{k+1,\ldots,n\}$, for some $k$.

For $n<6$ all peak-pit domains are trivially bipartite, since, if $D$ is not single-dipped, we can take $A$ to be any triple which has a never bottom condition, and then $\bar{A}$ will have size less than 3 and hence  define a single-dipped domain in a trivial way.  However, for $n\geq 6$ this is a non-trivial sub-class of the peak-pit domains.  For $n=6$ there are 9939 maximal peak-pit domains \cite{n7paper} and a computational search shows that 124 of these are not bipartite. In Table \ref{fig:nb} we display the non-bipartite maximal peak-pit domain with the smallest number of orders.  For $n=7$ there are 1465680 peak-pit domains, and of those 34393 are not bipartite.
This is a slight increase in the proportion of non-bipartite domains from $n=6$, but only to about 2\% of all peak-pit domains. For very large $n$ we expect most peak-pit domains to not be bipartite.

The maximum Condorcet domain for $n=8$, of size 224, found in \cite{8Alternatives} is a bipartite peak-pit domain, with $A$ of size 4. In fact, up to $n=8$ all maximum Condorcet domains belong to this class.

\begin{table}[h]
\centering
\begin{tabular}{*{19}{c}}
\toprule
1 & 1 & 1 & 1 & 1 & 1 & 1 & 1 & 1 & 1 & 1 & 6 & 6 & 6 & 6 & 6 & 6 & 6 & 6 \\
2 & 2 & 2 & 2 & 2 & 2 & 5 & 5 & 5 & 5 & 6 & 1 & 5 & 5 & 5 & 5 & 5 & 5 & 5 \\
3 & 3 & 3 & 5 & 5 & 5 & 2 & 2 & 2 & 6 & 5 & 5 & 1 & 2 & 2 & 2 & 2 & 4 & 4 \\
4 & 5 & 5 & 3 & 3 & 6 & 3 & 3 & 6 & 2 & 2 & 2 & 2 & 1 & 3 & 3 & 4 & 2 & 3 \\
5 & 4 & 6 & 4 & 6 & 3 & 4 & 6 & 3 & 3 & 3 & 3 & 3 & 3 & 1 & 4 & 3 & 3 & 2 \\
6 & 6 & 4 & 6 & 4 & 4 & 6 & 4 & 4 & 4 & 4 & 4 & 4 & 4 & 4 & 1 & 1 & 1 & 1 \\
\bottomrule
\end{tabular}
\caption{The smallest non-bipartite peak-pit domain. The domain has maximum width.}\label{fig:nb}
\end{table}

In a bipartite peak-pit domain  we have no restrictions on the never conditions for triples which intersect both $A$ and $\bar{A}$.  A natural subclass, which contains many classical Condorcet domains, are those which have a societal axis and never conditions which are determined by a subset of that axis. 
\begin{definition}
    Given a linear order $>$, the societal axis, on the set of alternatives, a domain $D$ is a \emph{midpoint bipartite peak-pit domain} if there exist a subset $A$ of the alternatives such that any triple $(a,x,b)$, where $a>x>b$, satisfies a never bottom condition when $x \in A$ and a never top condition when $x\in \bar{A}$. 
\end{definition}
The restriction of the domain $D$ to $A$ is an Arrow's  single-peaked domain and the restriction to $\bar{A}$ is the dual of an Arrow's single-peaked domain. Hence a midpoint bipartite peak-pit domain is indeed a bipartite peak-pit domain. Given a midpoint bipartite peak-pit domain $D$, its restriction to a subset of the alternatives is also a midpoint bipartite peak-pit domain.

The best-known examples of midpoint bipartite peak-pit domains are those defined by Fishburn's alternating scheme. There the sets $A$ and $\bar{A}$ are given by the even and odd integers respectively. The generalized Fishburn domains studied in \citep{Karpov2023, slinko2023} are also midpoint bipartite peak-pit domains. 

The smallest examples of maximal peak-pit Condorcet domain which are not midpoint bipartite peak-pit domain have four alternatives. One such example is the single-crossing domain \citep{slinko2021} on four alternatives. In this domain there is one never-top triple and three never-bottom (or dual of it), giving two triples with different never conditions but the same midpoint.

Note that unlike the bipartite peak-pit domain, this class is not closed under isomorphism, since an isomorphism typically changes the societal axis in a way which means that midpoints for triples are not preserved.  The dual of a midpoint bipartite peak-pit domain will however be a midpoint bipartite peak-pit domain.

For $n\leq 7$ we have used the data from \cite{n7paper}  to find all maximal peak-pit domains which are mid-point bipartite. The results are displayed in Table \ref{fig:mbb}.
\begin{table}[h]
\centering

\begin{tabular}{lllllllll}
\toprule
$n$& Total & \multicolumn{7}{l}{$(s, N )$} \\
\midrule
$4$ & 10 & (9, 1) & (8, 5) & & & & & \\
$5$ & 181 & (20, 2) & (19, 4) & (18, 2) & (16, 14) & & & \\
$6$& 9939& (45, 1)& (44, 4)& (42, 9)& (39, 16)& (38, 6)& (36, 3)& (32, 85)\\
$7$&1465680 &(100, 2)& (97, 4) & (96, 12) & (91, 4)& (89, 32)& (88, 6) & (87, 2) \\
& & (86, 16)& (84, 6)&(79, 128)&(78, 20)&(76, 8) & (72, 6)&(64, 1136) \\

\bottomrule
\end{tabular}
\caption{The number of midpoint bipartite maximal peak-pit domains. The column labelled Total gives the total number of maximal peak-pit domains for each $n$.  A pair $(s,N)$ means that there are $N$ midpoint bipartite domains of size $s$}\label{fig:mbb}
\end{table}
Note that the midpoint bipartite maximal peak-pit domains all have size at least $2^{n-1}$ for $n\leq 7$. We conjecture that this is true for all $n$.
\begin{conj}
    Each midpoint bipartite maximal peak-pit domain on $n$ alternatives has size at least $2^{n-1}$.
\end{conj}
In section 4 we will show that this lower bound is true for the domains generated by set-alternating schemes, a strict subclass of the midpoint bipartite peak-pit domains.

Up to $n=7$, the maximum size Condorcet domains for each $n$ are  midpoint bipartite peak-pit domains. However, the maximum size Condorcet domain for $n=8$, which is a bipartite peak-pit domain, is not a midpoint bipartite domain\footnote{This has been checked computationally.}, giving the first example of a maximum Condorcet domain which does not belong to this class.

In our next section we will proceed to study set-alternating schemes. These define a class midpoint bipartite peak-pit domains which use only two distinct never conditions but, as we shall show, for large $n$ they nonetheless produce the largest known Condorcet domains.

\section{Set-alternating schemes}
\label{sec:set_alter}

Our focus in this section is a class of domains which have a societal axis and relative to that axis are defined using only the never conditions $1N3$ and $3N1$. As we will see our particular class of such domains are mathematically natural, but the general class of domains using only these never conditions also have a natural social choice interpretation. Given a societal axis the never conditions $1N3$ and $3N1$  can both be seen as expressing a weak form of agreement with the ranking on the axis. The condition $1N3$ for a triple $a>b>c$ implies that the domain never ranks $a$ last of the three, and $3N1$ means that $c$ is never placed first.  Consequently, this can be used to model a relatively homogeneous society where the population mostly agree with a common societal axis. There  are some exceptional domains in the class. If we have only $1N3$ for all triples we get a unique Arrow's single-peaked domains, which is minimally rich, with only two alternatives ranked last. Dually, the domain with only $3N1$ has only two top alternatives and all alternatives are ranked last in some order.  However, for  all other domains in this class the set of top ranked alternatives, and last ranked alternatives, are both strict subsets of the alternatives, often quite small.  So, in this type of  homogeneous societies, or domains,  there is little diversity among the highest and lowest ranked alternatives. Strikingly, as our results will show, in terms of domain size some of these seemingly austere societies provide a larger variety of opinions than the single-peaked domains. Though here the source of variation is a diversity of alternatives ranked near the middle of the societal axis, rather than a variety of potential top or bottom alternatives. 

Let us start our analysis by defining the class of set-alternating schemes. 
\begin{definition}
Starting with a subset $A \subseteq [n]$ we consider the following  never conditions on triples $L(X)$: For $i < j < k$ with $j \in A$ assign the never condition $1N3$. For $i < j < k$ with $j \notin A$ we assign the never condition $3N1$. This is the \emph{set-alternating scheme} generated by $A$.

We let $D_X(A)$ denote the Condorcet domain which is generated by this scheme and let $f_n(A)$ denote the cardinality of $D_X(A)$.  
\end{definition}

Given that our definition of set-alternating schemes could easily be modified to use any other pair of never conditions, one may ask why we focus on the particular pair $1N3, 3N1$? If one tests out all pairs of peak-pit never conditions on all sets $A$ for some small $n$, say $n=5$,  one finds that most pairs do not generate a copious domain for every $A$, but a small set of pairs do. These pairs fall into four families.

The first are those which assign the same never condition to every triple. These domains have size $2^{n-1}$ by~\cite{raynaud1981paradoxical} and the domain does not depend on $A$, but it does depend on the choice of never condition.  The second family are those which use two never-bottom conditions, or two never-top conditions. This is a subset of the class of Arrow's single-peaked domains, and their duals.  These also have size $2^{n-1}$ but here the domain does depend on $A$. The class of all Arrow's single-peaked domains has been studied in detail in~\cite{slinko2019condorcet}. The third family uses the pair $2N1/2N3$. If we use this pair and  take $A$ to be the set of odd numbers, or even numbers, we recover Fishburn's alternating scheme. Generalised versions of Fishburn's domains are discussed in \cite{Karpov2023}).  The fourth family, which we will show to be copious for all $n$, is given by the pair $1N3,3N1$ and is the focus of the remainder of the paper.

Let us consider some examples. If $A=[n]$, then all triples are assigned the $1N3$ never condition. This gives an Arrow's single-peaked domain, but it is not Black's single-peaked since the domain does not contain two mutually reversed orders. This is the only Arrow's single-peaked domain given  by a set-alternating scheme.  By \cite{slinko2019condorcet} the cardinality of $D_X([n])$ is $2^{n-1}$.  For $A=\emptyset$ we get a domain which can be related to that for the empty set via the following lemma.
\begin{lemma}\label{compl}
    Define the \emph{reverse complement} $A^*$ of $A$ to be $A^*=n+1-\overline{A}$, where the subtraction means the set of numbers $n+1-i$ for all $i\in \overline{A}$.

    Let $C_1$ be the domain defined by the set $A$ and $C_2$ the domain given by $A^*$ then $C_2$ is obtained from $C_1$ by first reversing all orders and then reversing the names of the alternatives.
\end{lemma}
\begin{proof}
    Assume that we have a triple $(a,b,c)$ satisfying a never condition $xNp$, in Fishburn's notation. 
    
    Note that when all orders are replaced by their dual, every never condition $xNp$ is replaced by $xN(4-p)$.   
    
    Similarly, when the list of names is reversed a never-condition $xNp$ is replaced by $(4-x)Np$. Here the triple is  transformed to $(n+1-c,n+1-b,n+1-a)$.

    So the joint reversal leads  replaces $3N1$ by $1N3$, and replaces $(a,b,c)$ by $(n+1-c,n+1-b,n+1-a)$.
\end{proof}

For $n=4$, $D_X(\{2\})$ is isomorphic to Fishburn's alternating scheme. $D_X(\{2\})$ leads to a $1N3$ never condition for triples 1,2,3 and 1,2,4, and a $3N1$ never condition for triples 1,3,4 and 2,3,4. The same set of never conditions follows from $2N3$ never condition in case of median 1 and $2N1$ never condition in case of median 4 for ordering of alternatives 2143. For $n>4$ there is no $A$ corresponding to Fishburn's alternating scheme. 

Here we may ask which domain sizes set-alternating schemes more generally generate.  In Figure~\ref{fig1} we plot the pair (set size, domain size) for all subsets of $\{2,\ldots,n-1\}$ for $n=8$ and 9.  We exclude 1 and $n$ from $A$ since they cannot be midpoints of a triple and hence do not affect which domain we generate.  As we can expect from Lemma~\ref{compl} the  size distribution is symmetrical.  The lowest domain sizes in the figure are of the form $2^{n-1}$ and in Corollary~\ref{lowb}  we show that this is the minimum possible size.  We also see many schemes which lead to much larger domains and we can identify the maximum ones.
\begin{figure}
	\centering
	\begin{subfigure}{0.49\linewidth}
	\centering
	 \includegraphics[width=0.69\textwidth]{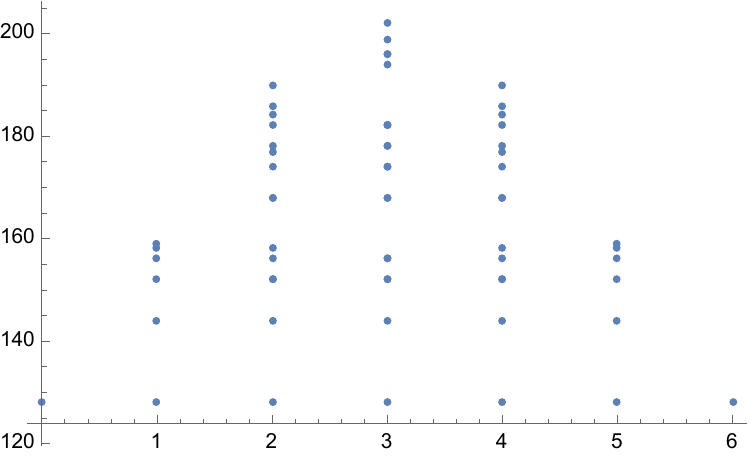}
	\caption{$n=8$}

	\end{subfigure}
  \hfill
	\begin{subfigure}{0.49\linewidth}
	\centering
	 \includegraphics[width=0.69\textwidth]{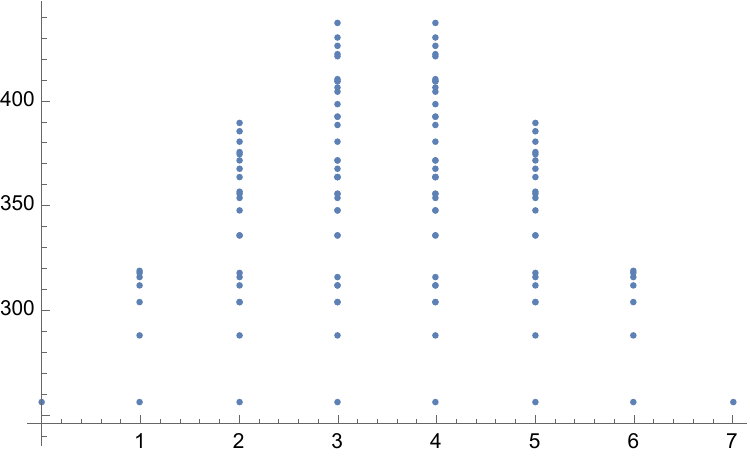}
\caption{$n=9$}
	\end{subfigure}
 \caption{Domain size and set size for all subsets of $\{2,\ldots,n-1\}$}

 \label{fig1}
\end{figure}

\begin{definition}
    $D_X(B_n)$ is the domain given by the \emph{odd  $1N33N1$-alternating scheme} if $$B_n=\{2,3,5,\ldots ,n-3+p_n\},$$ where $p_n=( n \bmod 2)$.
\end{definition}
This scheme produces the largest domain sizes for small $n$ and we will later investigate their growth rate.  For even $n$ $B_n$ is equal to its reverse complement. For odd $n$ the reverse complement consists of the even numbers $B_n^*=\{2,\ldots n-3\}$, which by Lemma~\ref{compl} generates a domain of the same size as $B_n$. If we extend  this  set of even numbers by $n-1$ we get a new scheme which we will use in  our asymptotic analysis. 
\begin{definition} 
    $D_X(A_n)$ is the result of the \emph{even $1N33N1$-alternating scheme} if $$A_n=\{2,4,6,\ldots ,n-2+p_n\},$$ where $p_n=( n \bmod 2)$.
\end{definition}
We will refer to the scheme given  by taking $B_n^*$ for odd $n$ and $A_n\setminus (n-2)$ for even $n$ as the \emph{truncated even $1N33N1$-alternating scheme}.

\begin{figure}[ht]
    \centering\small
    \begin{tikzpicture}[scale=1.]
    
    \centering
    \node[] (132564) at (0,0) {\textbf{132564}};
    \node[] (132546) at (-1,-1) {\textbf{132546}};
    \node[] (135246) at (1,-1) {\textbf{135246}};
    \node[label={[label distance=-1.6cm]0:\textbf{135264}}] (135264) at (2,0) {\phantom{1}};

    \node[] (312564) at (0,-2) {\textbf{312564}};
    \node[] (312546) at (-1,-3) {\textbf{312546}};
    \node[] (315246) at (1,-3) {\textbf{315246}};
    \node[label={[label distance=-1.6cm]0:\textbf{315264}}] (315264) at (2,-2) {\phantom{1}};

    \node[] (351246) at (1,-5) {\textbf{351246}};
    \node[label={[label distance=-0.1cm]-3:\textbf{351264}}] (351264) at (2,-4) {\phantom{1}};

    \node[] (135624) at (3,0.68) {\textbf{135624}};
    \node[] (315624) at (3,-1.32) {\textbf{315624}};
    \node[label={[label distance=-0.3cm]-0.1:\textbf{351624}}] (351624) at (3.,-3.85) {\phantom{1}};
    \node[] (356124) at (4,-2.83) {\textbf{356124}};

    \node[] (134265) at (-5,-1) {\textbf{134256}};
    \node[] (132465) at (-4,0) {\textbf{134265}};
    \node[] (132456) at (-2,0) {\textbf{132465}};
    \node[] (134256) at (-3,-1) {\textbf{132456}};
    \node[] (314265) at (-4,-2) {\textbf{314265}};
    \node[] (312465) at (-2,-2) {\textbf{312465}};
    \node[] (312456) at (-3,-3) {\textbf{312456}};
    \node[] (314256) at (-5,-3) {\textbf{314256}};
    \node[] (341256) at (-5,-5) {\textbf{341256}};
    \node[] (341265) at (-4,-4) {\textbf{341265}};

    \node[] (214356) at (-5,3) {\textbf{214356}};
    \node[] (214365) at (-4,4) {\textbf{214365}};
    \node[] (213465) at (-2,4) {\textbf{213465}};
    \node[] (213456) at (-3,3) {\textbf{213456}};
    \node[] (124356) at (-5,1) {\textbf{124356}};
    \node[] (124365) at (-4,2) {\textbf{124365}};
    \node[] (123465) at (-2,2) {\textbf{123465}};
    \node[] (123456) at (-3,1) {\textbf{123456}};

    \node[] (213546) at (-1,3) {\textbf{213546}};
    \node[] (213564) at (0,4) {\textbf{213564}};
    \node[label={[label distance=-1.6cm]-1.5:\textbf{215364}}] (215364) at (2,4) {\phantom{1}};
    \node[] (215346) at (1,3) {\textbf{215346}};
    \node[] (123546) at (-1,1) {\textbf{123546}};
    \node[] (123564) at (0,2) {\textbf{123564}};
    \node[label={[label distance=-1.6cm]-0.1:\textbf{125364}}] (125364) at (2,2) {\phantom{1}};
    \node[] (125346) at (1,1) {\textbf{125346}};

    \node[] (215634) at (3,4.7) {\textbf{215634}};
    \node[] (125634) at (3,2.7) {\textbf{125634}};

    \draw[] (132564.south) -- (132546.north);
    \draw[] (132564.south) -- (135264.south);
    \draw[] (132546.north) -- (135246.north);
    \draw[] (135246.north) -- (135264.south);

    \draw[] (312564.south) -- (312546.north);
    \draw[] (312564.south) -- (315264.south);
    \draw[] (312546.north) -- (315246.north);
    \draw[] (315246.north) -- (315264.south);

    \draw[] (132546.south) -- (312546.north);
    \draw[] (135246.south) -- (315246.north);
    \draw[] (135264.south) -- (315264.south);
    \draw[] (132564.south) -- (312564.north);

    \draw[] (315246.south) -- (351246.north);
    \draw[] (315264.south) -- (351264.south);
    \draw[] (351246.north) -- (351264.south);

    \draw[] (135624.south) -- (315624.north);
    \draw[] (315624.south) -- (351624.north);
    \draw[] (135264.south) -- (135624.south);
    \draw[] (315264.south) -- (315624.south);
    \draw[] (351264.south) -- (351624.north);
    \draw[] (351624.north) -- (356124.south);

    \draw[] (314265.south) -- (314256.north);
    \draw[] (314265.south) -- (312465.south);
    \draw[] (314256.north) -- (312456.north);
    \draw[] (312456.north) -- (312465.south);

    \draw[] (132465.south) -- (314265.north);
    \draw[] (134265.south) -- (314256.north);
    \draw[] (134256.south) -- (312456.north);
    \draw[] (132456.south) -- (312465.north);
    \draw[] (314265.south) -- (341265.north);
    \draw[] (314256.south) -- (341256.north);
    \draw[] (341265.south) -- (341256.north);

    \draw[] (132465.south) -- (134265.north);
    \draw[] (132465.south) -- (132456.south);
    \draw[] (134265.north) -- (134256.north);
    \draw[] (134256.north) -- (132456.south);

    \draw[] (134256.north) -- (132546.north);
    \draw[] (312456.north) -- (312546.north);

    \draw[] (214365.south) -- (214356.north);
    \draw[] (214365.south) -- (213465.south);
    \draw[] (214356.north) -- (213456.north);
    \draw[] (213456.north) -- (213465.south);

    \draw[] (124365.south) -- (124356.north);
    \draw[] (124365.south) -- (123465.south);
    \draw[] (124356.north) -- (123456.north);
    \draw[] (123456.north) -- (123465.south);

    \draw[] (213564.south) -- (213546.north);
    \draw[] (213564.south) -- (215364.south);
    \draw[] (213546.north) -- (215346.north);
    \draw[] (215346.north) -- (215364.south);

    \draw[] (123564.south) -- (123546.north);
    \draw[] (123564.south) -- (125364.south);
    \draw[] (123546.north) -- (125346.north);
    \draw[] (125346.north) -- (125364.south);

    \draw[] (214365.south) -- (124365.north);
    \draw[] (214356.south) -- (124356.north);
    \draw[] (213465.south) -- (123465.north);
    \draw[] (213456.south) -- (123456.north);

    \draw[] (213564.south) -- (123564.north);
    \draw[] (213546.south) -- (123546.north);
    \draw[] (215364.south) -- (125364.south);
    \draw[] (215346.south) -- (125346.north);

    \draw[] (213456.north) -- (213546.north);
    \draw[] (123456.north) -- (123546.north);

    \draw[] (123465.south) -- (132456.north);
    \draw[] (123456.south) -- (134256.north);
    \draw[] (123564.south) -- (132564.north);
    \draw[] (123546.south) -- (132546.north);

    \draw[] (215364.south) -- (215634.south);
    \draw[] (125364.south) -- (125634.south);
    \draw[] (215634.south) -- (125634.north);

    \end{tikzpicture}
    \caption{The median graph for the even $1N33N1$-alternating scheme for $n=6$.}
    \label{fig:median_graph}
\end{figure}
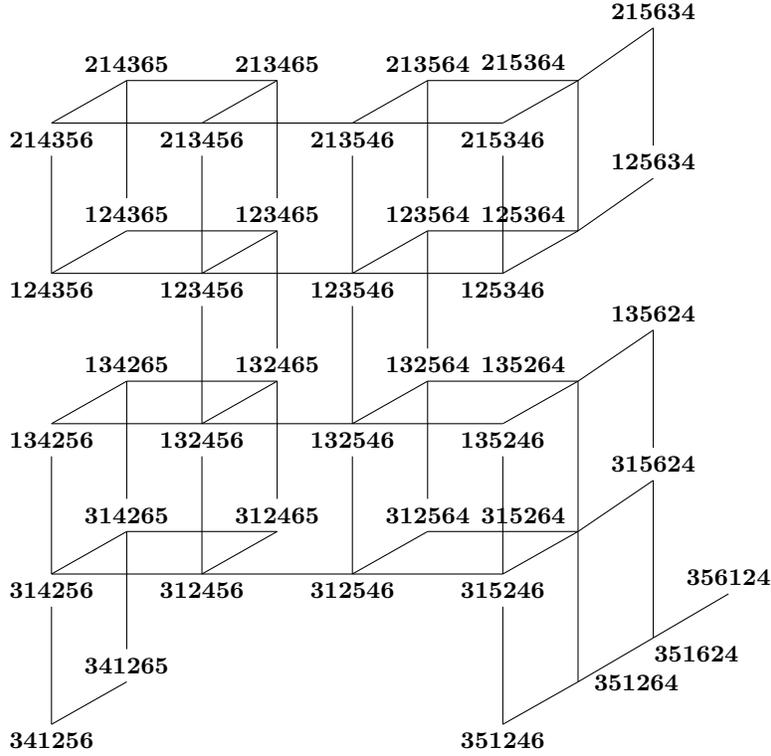

Figure~\ref{fig:median_graph} displays the median graph associated to the even $1N33N1$-alternating scheme for $n=6$.

In Table \ref{tab:size} we show the sizes of the domains generated by these schemes and Fishburn's alternating scheme, each computed using the Condorcet Domain Library (CDL) from \cite{zhou2023new, zhou2023cdl}.

\begin{table}[ht]
    \centering\small
\captionsetup{width=0.7\linewidth}
    \caption{Domain sizes for the \textit{odd} 1N33N1 scheme, the \textit{even} 1N33N1 scheme, and the alternating scheme. The \textit{odd} 1N33N1 scheme surpasses the alternating scheme for $n \ge 16$. }
    \begin{tabular}{ccccc}
    \toprule
    $n$ & \textit{odd} 1N33N1 & \textit{even} 1N33N1 & truncated \textit{even} 1N33N1 &  Fishburn's  alternating scheme\\
    \midrule
    4 & 9 & 9 & 8 & 9 \\
    5 & 19 & 18 & 19 & 20 \\
    6 & 42 & 42 & 39 & 45 \\
    7 & 91 & 84 & 91 & 100 \\
    8 & 202 & 199 & 190 & 222 \\
    9 & 437 & 398 & 437 & 488 \\
    10 & 973 & 950 & 922 & 1069 \\
    11 & 2102 & 1900 & 2102 & 2324 \\
    12 & 4690 & 4554 & 4464 & 5034 \\
    13 & 10122 & 9108 & 10122 & 10840 \\
    14 & 22617 & 21884 & 21587 & 23266 \\
    15 & 48779 & 43768 & 48779 & 49704 \\
    16 & \textbf{109104} & 105323 & 104322 & 105884 \\
    17 & \textbf{235197} & 210646 & 235197 & 224720 \\
    18 & \textbf{526441} & 507398 & 503966 & 475773 \\
    19 & \textbf{1134474} & 1014796 & 1134474 & 1004212 \\
    20 & \textbf{2540586} & 2446022 & 2434088 & 2115186 \\
    \bottomrule
        
    \end{tabular}
    \label{tab:size}
\end{table}

Regarding the maximum domain size we conjecture the following.

\begin{conj} 
    For each $n \geq 5$, the maximum size of set-alternating domain is the size of the odd $1N33N1$-alternating domain.
\end{conj}
This conjecture has been verified computationally for $n\leq 24$. As we can see for $n<16$ Fishburn's alternating scheme leads to larger domains than the odd $1N33N1$-alternating  scheme.

\begin{prop}\label{prop:maximal}
    Each domain defined by a set-alternating scheme is a  copious peak-pit maximal Condorcet domain.
\end{prop}
\begin{proof}
For $n\leq 3$ each maximal set-alternating domain is a maximal copious peak-pit  Condorcet domain.  Let us assume that it is true for $n-1$ alternatives.

First, let us recall that in a copious domain every triple of alternatives satisfies exactly one never condition. 

Suppose there is a maximal set-alternating domain $C$ defined by set $A$ with corresponding $1N33N1$ never-conditions that is not a maximal Condorcet domain. If this is the case then there is a domain $C'\supset C$ that satisfies a set of never-conditions which  does not coincide with the set-alternating conditions. Thus at least one triple initially satisfies a pair of never-conditions, the one from the set-alternating definition and one given by $C'$. This leads to a contradiction if $C$ is copious.

Let $C_1$ be the restriction of $C$ to $[n]\setminus 1$. The domain $C_1$ is also a set-alternating domain and hence copious by our induction assumption.  If we add 1 as the highest ranked element in each order of $C_1$  we get a new domain $C_1+$ on the set of alternatives $[n]$.  Each order in $C_1+$ is compatible with  both $1N3$ and $3N1$ never conditions for triples $1,i,j$, and all other triples satisfy the never condition given by $A$. Hence $C_1+ \subset C$, and  since the restriction of $C_1+$ to any triple not containing 1 has size four this is true for $C$ as well. This shows that any triple not containing 1 satisfies the condition required for $C$ to be copious. 

We can repeat this argument for the restriction of $C$ to $[n]\setminus n$, getting a subdomain $C_n+$, where $n$ is ranked last in every order, which shows that every triple not containing $n$ has a restriction of size 4.  It remains to show that triples of the form $1,j,n$ also give restrictions of size 4.

Since  $C_1$ is copious, and hence ample, $C_1+$ restricted to a triple $1,j,n$ contains both $1jn$ and $1nj$. Similarly $C_n+$ contains both $1jn$ and $j1n$. So the restriction of $C$ to this triple contains $1jn,1nj,j1n$.

The domain $C$ contains orders which rank the alternatives from $\overline{A}$ highest, then $n1$, and finally all alternatives from $A$. If $j\in A$, then the restriction of $C$ to $1,j,n$ contains $n1j$, giving a fourth order and the triples satisfy only the never condition $1N3$. If $j\notin A$, then the restriction of $C$ to $1,j,n$ contains $jn1$, giving a fourth order and the triple satisfies only the never condition $3N1$. 

So, we have shown that the domain $C$ is copious, and it is thereby also a maximal Condorcet domain.
\end{proof}

In addition to potentially large sizes, the set-alternating schemes generate domains with good local structure.

\begin{prop}\label{prop:con}
    Each domain defined by a set-alternating scheme is connected.
\end{prop}

\begin{proof}
For small $n\leq 3$  domains defined by a set-alternating scheme are connected.  Let us assume that this also is true for $n-1$ alternatives.

Let a domain $C$ with $n$ alternatives be defined by the set-alternating scheme of the set $A$. For each order $a_1a_2\ldots a_n\in C$ such that $a_k=n$ orders $a_1\ldots a_{k-1}na_{k+1} \ldots a_n$, $a_1\ldots a_{k-1}a_{k+1}na_{k+2}\ldots  a_n$, $a_1\ldots a_{k-1}a_{k+1}\ldots  a_nn$ belong to the domain $C$. Thus each order from $C$ is connected with an order from $C$ that has $n$ ranked last. The set of orders with $n$ last constitute a domain that can be obtained from the domain $C'$ on $[n-1]$ defined by the set-alternating scheme of the set $A\setminus n$ by concatenating $n$ to each order. Because of the induction hypothesis the domain $C'$ is connected. Thus the domain $C$ is connected. 
\end{proof}

All set-alternating domains have a well-structured common part. The Condorcet domain which for each triple satisfies the two never conditions 1N3 and 3N1 is a subset of every set-alternating domain. For unitary domains the common part has the following structure. Each order from the common part is obtained by swaps of  disjoint pairs of neighbouring alternatives in order $12\ldots n$. For each triple of neighbouring alternatives $x-1,x,x+1$ we have orders $x-1xx+1,xx-1x+1,x-1x+1x$ as the  restriction to the set $\{x-1,x,x+1\}$.

\begin{prop}\label{prop:Fibonacci}
    The size of the domain on $n$ alternatives which  for each triple satisfies never conditions 1N3 and 3N1, is the $(n+1)^{th}$ Fibonacci number.
\end{prop}	
\begin{proof}
The set of orders from the considered domain can partitioned into orders with last $n$ and last $nn-1$. The size of the first part is the size of the considered domain type for $n-1$ alternatives, and the size of the first part is the size of the considered domain type for $n-2$ alternatives. Thus, we have $c_n=c_{n-1}+c_{n-2}$. Starting from $c_1=1$, $c_2=2$ we get a shifted Fibonacci sequence.
\end{proof}
The asymptotic growth rate of the Fibonacci sequence is $\frac{1+\sqrt{5}}{2}\approx 1.618$.

Set-alternating schemes can be used as a straightforward way of producing random examples of maximal Condorcet domains.  For a given $n$ we can construct a set $A$ at random by including each value from $X_n$ with probability $p$, independently, and then construct a domain $D$ by using the set-alternating scheme given by $A$. For $p=\frac{1}{2}$ we get the uniform distribution for $A$. 
\begin{quest}
    What is the expected size of the domain $D$ for the set-alternating scheme of a random set $A$, when $A$ is generated with inclusion probability $p$?
\end{quest}
Since the domain size here is exponential in $n$ the expected size might deviate significantly from the median size. The same question could be asked for the median. 

\subsection{Recursive properties of the domain sizes}
First of all let us note that whether either of 1 and $n$ is a member of $A$ or not does not affect the size of the generated domain. Simply because neither of these can be the middle element of a triple. Next, let $w$ denote the largest element in $A$, and $A'=A\setminus w$. 

\begin{prop}
    If $w=n-1$, then $f_n(A)=2f_{n-1}(A')$.
\label{double}
\end{prop}
\begin{proof}
    For all triples $i,n-1,n$, we have the never condition $1N3$. Thus, only alternatives $n-1$ and $n$ can be last in orders from $D_X(A)$. Note that adding $n$, or $n-1$ at the end of an order, which does not use that alternative, does not violate 1N3, 3N1 conditions.  Hence, the number of linear orders with $n$ last equals $f_{n-1}(A')$, and the same holds for those with $n-1$ last.
\end{proof}

\begin{prop}
    If $1<w<n-1$, then $f_n(A)=2f_{n-1}(A)+f_{n-1}(A')-S$, where $S=\sum_{j=w-1}^{n-2}f_j(A')$.
\label{general}
\end{prop}
\begin{proof}
    For all triples $i<n-3<j$ we have the restriction $1N3$. Thus, only alternatives $w,\ldots,n$ can be last in $D_X(A)$.
    
    There are $f_{n-1}(A)$ orders, in which $n$ is the last alternative, and there are $f_{n-1}(A)$ orders  in which $n$ is the second last alternative. 
 
    If $n$ is neither last nor second last, then alternative $w$ is last.

    There are $f_{n-1}(A')$ orders that satisfy all conditions without alternative $w$. Adding $w$ at the end  may violate some 3N1 conditions in some of the orders. We should subtract the number of orders with $nw$ last (we calculate them when finding the number of orders with the second last $n$), the number of orders which end with $nn-1w$, the number of orders which end with $n-1nn-2w$,..., $w+2\ldots n-1nw+1w$. Thus we have $f_{n-1}(A')-\sum_{j=w-1}^{n-2}f_j(A')$ orders from $D_X(A)$  with $w$ last. 

    Summing we obtain the result. 
\end{proof}

By Proposition~\ref{general} we have  $f_n(A)\geq2f_{n-1}(A)$, giving us a general lower bound on the size of set-alternating domains.
\begin{corollary}\label{lowb}
    We have $f_n(A)\geq 2f_{n-1}(A)$ and $f_n(A)\geq2^{n-1}$.
\end{corollary}
\begin{proof}
    We have $f_n(A)=2f_{n-1}(A)+f_{n-1}(A')-\sum_{j=w-1}^{n-2}f_j(A')$, where $2f_{n-1}(A)$ is the number of orders, in which $n$ is last or second last and $f_{n-1}(A')-\sum_{j=w-1}^{n-2}f_j(A')$ is the number of orders, in which $n$ is neither last, nor second last. Thus, we have $f_{n-1}(A')-\sum_{j=w-1}^{n-2}f_j(A')\geq 0$ and $f_n(A)\geq 2f_{n-1}(A)$. Starting from $f_1=1$ we get $f_n(A)\geq2^{n-1}$.
\end{proof}

We also get some information about sets which generate large domains.
\begin{corollary}\label{incl2}
    If $n-1$ is a member of $A$ then $A_1=A\setminus n-1$
    generates a domain which is at least as large as that for $A$.
    
    If 2 is not a member of $A$ then $A_2=A\cup\{2\}$ generates a domain which is at least as large as that for $A$.   
\end{corollary}
\begin{proof}
    Let us prove the first part. First note that $A_1=A'$.
    We have $f_n(A)=2f_{n-1}(A_1)$ and $f_n(A_1)=2f_{n-1}(A_1)+f_{n-1}(A_1')-\sum_{j=w-1}^{n-2}f_j(A_1')$. Thus, we have $f_n(A)\leq f_n(A_1)$.

    For the second part, consider the reverse complement $A^*$, which generates a domain of the same size as $A$. If 2 is not a member of $A$ then $n-1\in A^*$ and as we just showed the domain would have been at least as large if $n-1$ had not been a member of $A^*$.
\end{proof}

Hence the set of largest, for each $n$, set-alternating domains, contains domains such than $A$ contains 2 and does not contain $n-1$.

The following proposition determines the exact sizes for some special cases. 
\begin{prop} \label{prop:one_k}
    For $k=1$, we have $f_n(\{n-k\})=2^{n-1}$ and $f_n(\{k\})=2^{n-1}$. For $k>1$ and $n\leq k+1$, we have $f_n(\{n-k\})=2^{n-1}$ and $f_n(\{k\})=2^{n-1}$. For $k>1$ and $n>k+1$, we have $f_{n}(\{n-k\})=5\cdot 2^{n-3} - 2^{n-k-2}$ and $f_n(\{k\})=5\cdot 2^{n-3}-2^{k-2}$.
\end{prop}
\begin{proof}
For $k=1$, $f_n(\{n-k\})=f_n(\{k\})=2^{n-1}$ follows from Proposition~\ref{double}.

For $k>1$ and $n\leq k+1$, we have $f_n(\{n-k\})=f_n(\{k\})=2^{n-1}$ because it is the size of any Arrow's single-peaked domain.

From Proposition~\ref{general} for $k>1$ and $n>k+1$ we have $f_n(\{n-k\})=2f_{n-1}(\{n-k\}) + 2^{n-2}-\sum_{j=n-k-1}^{n-2}2^{j-1}$. Rearranging, we have
$f_{n}(\{n-k\})=2f_{n-1}(\{n-1-(k-1)\})+2^{n-k-2}$.

For $k=2$, we have $f_{n}(\{n-2\})=2^{n-1}+2^{n-4}$. For $k=3$, we have $f_{n}(\{n-3\})=2(2^{n-2}+2^{n-5})+2^{n-5}=2^{n-1}+3\cdot 2^{n-5}$. For $k=4$, we have $f_{n}(\{n-4\})=2(2^{n-2}+3\cdot 2^{n-6})+2^{n-6}=2^{n-1}+7 \cdot 2^{n-6}$. Solvingthe  recurrence, we have $f_{n}(\{n-k\})=2^{n-1}+(2^{k-1}-1) \cdot 2^{n-k-2}$.

From Proposition~\ref{general} for $n>k+1$ we have $f_n(\{k\})=2f_{n-1}(\{k\}) + 2^{n-2}-\sum_{j=k-1}^{n-2}2^{j-1}$. Having initial condition $f_{k+1}(\{k\})=2^{k}$ we get the solution $f_n(\{k\})=5\cdot 2^{n-3}-2^{k-2}$.
\end{proof}

Utilizing Proposition~\ref{general} once more one can find the size of the set-alternating domains with $|A|=2$. Iteratively applying Proposition~\ref{general} one can in principle find the size of any set-alternating domain. Another special case is solved in the following proposition.
\begin{prop} \label{prop:half_A1}
    For even $n\geq4$, we have $f_{n}(\{2,3,\ldots ,\frac{n}{2}\})= 2^{n-1}+\sum_{i=1}^{\frac{n}{2}-1}\sum_{j=1}^{\frac{n}{2}-1}\binom{i+j-2}{i-1}2^{n-i-j-2}$.
\end{prop}

\begin{proof}
The shuffle of orders $a\in L(A)$, $b\in L(B)$, where the sets $A$ and $B$ are disjoint, is an order $c\in L(A\cap B)$ such that the restriction of $c$ to $A$ equals $a$, and the  restriction of $c$ to $B$ equals $b$.

The domain $D_X(\{2,3,\ldots \frac{n}{2}\})$ contains orders, that have the structure $a+b+c$, where $+$ is concatenation $a\in D_{\{1,2,3,\ldots \frac{n}{2}-i-1\}}(\{2,3,\ldots \frac{n}{2}-i\})$, $b$ is a shuffle of orders $\frac{n}{2}-i+1\ldots \frac{n}{2}$ and $\frac{n}{2}+1\ldots \frac{n}{2}+j$, $c\in D_{\{\frac{n}{2}+j+1\ldots n-1\}}(\{\frac{n}{2}+j+1\ldots n\})$. This domain satisfies all never conditions that follow from the definition of the set-alternating domain. \cite{karpov2023constructing} proved that domains with this structure are a maximal Condorcet domains.

The number of orders in the shuffle of two orders with $i$ and $j$ alternatives is $\binom{i+j}{i}$. If the  first alternative belongs to the second order and the last alternative belongs to the first order, then we have $\binom{i+j-2}{i-1}$.

There are $2^{n-i-j-2}\binom{i+j-2}{i-1}$ linear orders which  have the described above structure and cannot be described by a smaller shuffle.

Thus, we have $$f_{n}(\{2,3,\ldots ,\frac{n}{2}\})= 2^{n-2}+\binom{n-2}{\frac{n}{2}-1}+2\sum_{i=1}^{\frac{n}{2}-1}\binom{i+\frac{n}{2}-2}{i-1}2^{\frac{n}{2}-i-1}+\sum_{i=1}^{\frac{n}{2}-1}\sum_{j=1}^{\frac{n}{2}-1}\binom{i+j-2}{i-1}2^{n-i-j-2}.$$ Simplifying, we obtain the result. \end{proof}

For $n=2,4,6,8,\ldots$ we have 2,9,42,194,884,3978,... The asymptotic growth rate of this sequence is 4, as a function of $n/2$. Including odd $n$ we get asymptotic growth rate 2, as a function of $n$.

Proposition~\ref{double} leads to the following proposition.
\begin{prop} 
    For $1<k<n$, we have $f_{n}(\{k,\ldots, n-1\})= 2^{n-1}$.
\end{prop}
In the considered cases the set $A$ consists of successive alternatives. All these cases are associated with relatively small domain sizes. Sets $A$ with a large number of alternations $i\in A,i+1\in \overline{A}$ possibly leads to a greater size. In the next section we will analyse one such case.  

\section{The asymptotic growth  rate of the set-alternating domains}
\label{sec:asymp}

Here we analyze even set-alternating domain $D_X(A_n)$ and its size $a(n)$. In proofs we use $A_n=\{1,2,4,6,\ldots ,n-2+p_n\}$, $p_n=( n \bmod 2)$. Adding 1 to set $A_n$ does not change the definition, but universalizes notation. From now on we let $a(n)$ denote the size of the domain resulting from the even $1N33N1$-alternating scheme. 

We partition all orders in  $D_X(A_n)$ on orders that start from set $\{1,2\}$, orders that start from set $\{1,2,3,4\}$, but not from set $\{1,2\}$,
orders that start from set $\{1,2,3,4,5,6\}$, but not from set $\{1,2,3,4\}$, etc. Orders from part $k$ start from $[2k]$, but not from $[2(k-1)]$.

Orders from the first part start from 12, 21. There are $2a(n-2)$ such orders.
Orders from the second part start from 1324, 1342, 3124, 3142, 3412. There are $5a(n-4)$ such orders. For $n=6$, 18 orders from the upper side of Figure~\ref{fig:median_graph} constitute the first part, 10 orders from the bottom left side of Figure~\ref{fig:median_graph} constitute the second part, 14 orders from the bottom right side of Figure~\ref{fig:median_graph} constitute the third part.

\begin{lemma}
    In all orders from $k^{th}$ part of $D_X(A_n)$ alternatives from $A_{2k}$ are in ascending order.
    \label{lem:2}
\end{lemma}
\begin{proof}

We will proof the statement by contradiction.

Suppose there is a pair of elements $i>j$ from $A_{2k}$ and an order $q$ from part $k$ of $D_X(A_n)$ such that $i,j$ are in descending order in $q$ and there is no $x\in A_{2k}$, $j<x<i$ such that $i,x$ are in descending order in $q$. We take the closest pair in descending order. Note that $i$ is even and $i<2k$.

Now we will prove that all alternatives $k\leq i$, $k\neq j$ precede $j$ in order $q$.

Suppose, that there is element $x\in[2k]$ between $i$ and $j$ in order $q$. If $j<x<i$, then this triple violates 1N3 and 3N1. If $x>i$, then $i,j,x$ violates 1N3. Thus, the only possibility is $x<j$.

If there is $y>i$ such that it precedes $i$ in order $q$, then  $i,j,y$ violates 1N3. Thus only elements below $i$ precede $i$ in order $q$.

If there is $z<j$ that stays after $j$ in order $q$, then $i,j,z$ violates 1N3. If there is $z$, $j<z<i$ that stays after $j$ in order $q$, then $z\in A_{2k}$, otherwise triple violates $3N1$. In this case pair $z,i$ is in descending order and closer, than pair $j,i$. There is no such $z$.

Thus, all alternatives $k\leq i$, $k\neq j$ precede $j$ in order $q$, and order $q$ belongs to part $i/2$, or lower. Because $2k\notin A_{2k}$ we have $i/2<k$. We get a contradiction.
\end{proof}

\begin{lemma}
    In all orders from $k^{th}$ part of $D_X(A_n)$ alternatives from $\overline{A_{2k}}$ are in ascending order.
\end{lemma}
\begin{proof}

Because of lemma~\ref{lem:2} elements from $A_{2k}$ are in ascending order in each order from part $k$.

Let us prove by induction that for each $j\in[k-1]$ element $2j+1$ precede $2j$ in each order from part $k$.

For $j=1$, elements $i\in\overline{A_{2k}}$ precede 2 in each order from part $k$. We have $i=3$, otherwise triple $2,3,i$ violates $3N1$. By induction hypothesis for each $l<j$, element $2l+1$ precedes $2l$ in each order from part $k$. let us prove for $j$. All $i<2j$ precede $2j$ in each order from part $k$. Elements $i\in\overline{A_{2k}}$, $i>2j$ precede $2j$ in each order from part $k$. We have $i=2j+1$, otherwise triple $2j,2j+1,i$ violates $3N1$.

Suppose there is a pair of elements $a>b$ from $\overline{A_{2k}}$ such that $a,b$ are in descending order in order $q$ from part $k$ of $D_X(A)$. Since $b$ precede $b-1$ in order $q$, triple $b-1,b,a$ violates $3N1$. We get a contradiction.
\end{proof}

\begin{definition}
    A sequence $a_1a_2\ldots a_{2k}$ of $k$ elements $u$ and $k$ elements $d$ such that for all $1\leq j \leq 2k$ we have $|i\in[j]|a_i=u\}|\geq|i\in[j]|a_i=d\}|$  is a \emph{Dyck word}.
\end{definition}

\begin{prop} \citep{Deutsch}
The number of Dyck words of size $2k$ is $C_k$, where $C_k$ is the $k^{th}$ Catalan number. \end{prop}

Let us define a bijection $\mu$ between with Dyck words of length $2(k+1)$ and the top $2k$ elements segment of orders in $k^{th}$ part of $D_X(A_n)$.

\begin{definition}
   Our bijection produces Dyck words as follows.
   
   The first element in Dyck words is $u$. It has no correspondence in orders. Each consequent element in the top $2k$ elements segment of an order from $k^{th}$ part of $D_X(A_n)$ corresponds with the consequent element in Dyck word: if the element belongs to $A_{2k}$ then $d$, if not, then $u$. The last element in Dyck word is $d$. It has no correspondence in orders. The length of the constructed Dyck word is $2(k+1)$.
\end{definition}

\begin{table}[ht]
\caption{An example}
    \label{tab:dyck}
    \centering
    \begin{tabular}{cc}
    \toprule

    Top 4 elements & Dyck word \\
    \midrule
        1324 & $ududud$ \\
        1342 & $uduudd$ \\
        3124 & $uuddud$ \\
        3142 & $uududd$ \\
        3412 & $uuuddd$ \\
    \bottomrule
    \end{tabular}
    
\end{table}

Table~\ref{tab:dyck} presents the bijection applied the top four elements of the second part of  $D_X(A_n)$. The following lemmas prove that the bijection is well-defined.

\begin{lemma}
    The image of each Dyck word of length $2(k+1)$ is a top $2k$ elements segment of an order from $k^{th}$ part of $D_X(A_n)$. 
\end{lemma}
\begin{proof}
Let us consider nonboundary subwords $a_ia_ja_l,$ $i,j,l\in[2k]$ (boundary elements have indices $0$ and $2k+1$). There are eight types such subwords: $ddd$, $ddu$, $dud$, $duu$, $udd$, $udu$, $uud$, $uuu$. The corresponding order should satisfy 1N3 if the median is $d$, and 3N1 if the median is $u$. Because the order of elements in $A_{2k}$ and in $\overline{A_{2k}}$ are increasing, images of $ddd$, $dud$, $udu$, $uuu$ satisfy 1N3, 3N1, image of $udd$ satisfies 1N3, image of $uud$ satisfies 3N1.

Image of $ddu$ violates 1N3, if
$$|\{r\in [2k]|a_r=u, r<l\}|+2<|\{r\in [2k]|a_r=d, r<i\}|.$$
It is impossible because of the definition of the Dyck word.

Image of $duu$ violates 3N1, if 
$$|\{r\in [2k]|a_r=u, r<j\}|+2<|\{r\in [2k]|a_r=d, r<l\}|.$$
It is impossible because of the definition of the Dyck word.

For each $r<k$ we have $|[2r]  \cap A_{2k}|=2+|[2r] \cap \overline{A_{2k}}|$. The image of the first 
$2r$ of each Dyck word contains at least one $u$ that corresponds with an element from set $[2n]\setminus[2k]$. The corresponding order belongs to part $k$. \end{proof}

\begin{lemma}
    The image of each top $2k$ elements segment from $k^{th}$ part of $D_X(A_n)$ is a Dyck word.
\end{lemma}

\begin{proof}
Because $|A_{2k}|=|\overline{A_{2k}}|$ we have equal numbers of $u$ and $d$.

Suppose there is an order $q$ from $k^{th}$ part of $D_X(A_n)$ such that for corresponding Dyck work there is $0\leq j \leq 2k-1$ with $|\{i\in[j]\cup0|a_i=u\}|<|\{i\in[j]\cup0|a_i=d\}|$. Therefore there is $0\leq r \leq j$ such that $|\{i\in[r]\cup0|a_i=u\}|+1=|\{i\in[r]\cup0|a_i=d\}|$. Thus, order $q$ belongs to part $r/2$, or lower. We get a contradiction.
\end{proof}

We are ready to present our main result.

\begin{prop}
    For even $n$ we have $$a(n) \sim \frac{\sqrt2}{4}\left(\sqrt{2+2\sqrt{2}}\right)^n,$$ for odd $n$ we have $a(n) \sim \frac{\sqrt{\sqrt{2}-1}}{2}\left(\sqrt{2+2\sqrt{2}}\right)^n$.
\end{prop}
\begin{proof}
For $m=n/2$, even $n$ we define 
$w(m)=a(n)$. From bijection we have
\begin{equation}
    w(m)=\sum_{k=1}^m C_{k+1}w(m-k).
    \label{Catalan}
\end{equation}

Starting from $w(0)=1$ we have a recurrence over even numbers. After rearranging indices one  finds that this recurrence is  recurrence (32) from \cite{AW2010}.

The sequence $w(m)$ is sequence A289684 in  \cite{sloane2003line}. Robert Israel showed that $w(m)\sim \frac{\sqrt2}{4}\left(2+2\sqrt{2}\right)^m$. Having $m=n/2$ we get the result for even $n$. From proposition~\ref{double} for odd $n$ we have $a(n)=2a(n-1)$ that means $a(n) \sim \frac{1}{2\sqrt{1+\sqrt{2}}}\left(\sqrt{2+2\sqrt{2}}\right)^n$. \end{proof}

For even $n$, $a(n)$ is member $n/2$ of the  sequence A289684 in \cite{sloane2003line}. For $n=1,2,\dots$, we have $1,2,4,9,18,42,84,199,...$.

Because 1N33N1 alternating scheme produces a peak-pit domain we have $f(n),h(n)\in \Omega(\left(\sqrt{2+2\sqrt{2}}\right)^n)$, where $\sqrt{2+2\sqrt{2}}=2.197368...$. 

Equation~\ref{Catalan} has coefficients which are Catalan numbers. Catalan numbers have many interesting properties (see e.g. \cite{Stanley}). Proposition~\ref{odd} exploits one of the Catalan numbers properties.

\begin{prop}
    $a(n)$ is odd if and only if $n=2^t$, $t>1$, or $t=0$.
\label{odd}
\end{prop}
 
\begin{proof}
We will prove this by induction. The statement is true for $n\leq4$. Suppose the statement is true for $n\leq2^t$. Let us prove for $n\in\{2^t+1,\ldots, 2^{t+1}\}$.

If $n$ is odd, then by Proposition~\ref{double} we have $a(n)=2a(n-1)$, $a(n)$ is even. For even $n$ we have equation~\ref{Catalan}, where $m=n/2$. 

$C_{k+1}$ is odd if and only if $k+1=2^r-1$, $r>1$ \citep{Egecioglu}. $w(m-k)$ is odd if and only if $m-k=0$, or $m-k=2^s$, $s>0$ (it follows from from induction hypothesis). 

For $k=m$, both numbers are odd if $m=2^r-2$.

For $k=m-2^s$, both numbers are odd if $m=2^s+2^r-2$.

If we have $m=2^{r_1}-2$, then there is exactly one pair $r_2,s$ such that $2^s+2^{r_2}-2=2^{r_1}-2$ (solution is $s=r_2=r_1-1$). Thus for all $m=2^r-2$, $m>2$ we have two odd members in equation~\ref{Catalan}, the sum is even.

The expression $m=2^{s}+2^{r}-2$ is an equation with unknowns $s$ and $r$. The number $m+2$ has a unique partition into two powers of two, if it exists. If the parts are equal numbers, then $m=2^{r_1}-2$ for some $r_1$. In this case we have an even sum. If the parts are different and $s,r>1$, then there are two ways to define $s,r$, and the sum is even. If $s=1$ and $r>1$, then there is unique $k=m-2$ that leads to odd member of the sum. There is no other pair of $s,t$, because $r>1$. Thus, we have $m=2+2^r-2$ and the sum is odd. Having $r=t+1$ we obtain the result.
\end{proof}

The domain given by the even 1N33N1 alternating scheme exceeds the domains from Fishburn's alternating scheme in size for all $n\geq 18$, but as we have already seen it is not the largest set-alternating domain. The odd 1N33N1 alternating scheme leads to larger Condorcet domains. However, using Proposition~\ref{general} we can show that the ratio of their sizes is bounded by a constant factor. In particular the growth rate is of the same order as $a(n)$.
\begin{prop}
   The odd 1N33N1 alternating scheme generates domains of size  $\Omega(\left(\sqrt{2+2\sqrt{2}}\right)^n)$
   \label{prop:newA}
  \end{prop}
  \begin{proof}
      The restriction of $D_X(\{2,3,5,7,\dots,n-3\}$ to the set $X\setminus\{2,n-1\}$ contains $a(n-2)$ orders. The restriction of $D_X(\{2,3,5,7,..,n-2\}$ to the set $X\setminus\{2\}$ contains a(n-1) orders.

      The restriction of $D_x(A_n)$ to the set 
      $X\setminus\{3,n-2\}$ contains $f_{n-2}(\{2,3,5,7,\dots,n-5\})$ orders.  The restriction of $D_x(A_n)$ to the set $X\setminus\{3\}$ contains $f_{n-1}(\{2,3,5,7,\dots,n-3\})$ orders.

      Thus, $f_n(\{2,3,5,7,\dots,n-3\}),f_n(\{2,3,5,7,\dots,n-2\})$ are between $a(n-2)$ and $a(n+2)$ that proves the statement.
  \end{proof}


\section{Discussion}
\label{sec:disc}

In this paper we introduced and studied the class of bipartite peak-pit domains.  This class provides a common generalisation of single-peaked and single-dipped domains, based on a partition of the alternatives into two parts. The voters rank alternatives in one part with single-peaked preferences, and in the other part with single-dipped preferences, on a common societal axis. This captures situations where a voter has a most preferred location on the axis for alternatives in the first part and a least preferred location for alternatives in the second part, thereby mixing two different rationales for the ranking.
We have found that for small numbers of alternatives $n$ most peak-pit domains are bipartite, and for each $n\leq 8$ the largest peak-pit domain is bipartite. 

Next we considered the midpoint bipartite peak-pit domains. This is a subclass of the bipartite domains with additional restrictions on the never conditions for triples which contain alternatives from both of the two parts.  This class is significantly smaller than the full set of bipartite peak-pit domains, but conjecturally they are always at least as large as single-peaked domains and up to $n=7$ the largest peak-pit domains for each $n$ belong to this class. The latter is due to the fact that Fishburn's domains are midpoint bipartite for all $n$. 

Fishburn's alternating scheme has long been one of the cornerstones for constructing examples of large Condorcet domains. As we have shown, the family of set-alternating schemes can produce even larger domains, taking over from Fishburn's alternating scheme at $n=16$, just  like the examples from Fishburn's replacement scheme did. 

Our paper leaves open conjectures and questions for each of the domain classes we have studied. We believe that the bipartite  peak-pit domains are of interest both for direct applications, where they model preferences with mixed rationales, and for the more mathematical study of large Condorcet domains.

\section*{Acknowledgements}
The authors would like to thank Arkadii Slinko for his useful comments on the manuscript. The work on set-alternating schemes was motivated by some unpublished research notes by Dr Riis. Bei Zhou was funded by the Chinese Scholarship Council (CSC). The Basic Research Program of the National Research University Higher School of Economics partially supported Alexander Karpov. This research utilised Queen Mary's Apocrita HPC facility, supported by QMUL Research-IT.

\end{document}